\documentclass[conference]{IEEEtran}
\usepackage{hyperref}
\hypersetup{
  hypertexnames=false,
  colorlinks=true,
  linkcolor={blue!60!black},
  citecolor={blue!60!black},
  urlcolor={blue!60!black}
}

\usepackage{booktabs}
\usepackage{color}
\usepackage{amsthm}
\usepackage{amsmath,amsfonts} 
\usepackage{url} \usepackage{graphicx}
\usepackage{mathtools}

\usepackage{footnote}
\usepackage{framed}
\usepackage{xspace}
\usepackage{multirow}
\usepackage{enumitem}

\usepackage{tikz}
\usetikzlibrary{calc,positioning,shapes,shadows,arrows,fit}

\newcommand{\ignore}[1]{}
\newcommand{\rev}[1]{#1}
\newcommand{\newrev}[1]{#1}

\usepackage{caption}
\numberwithin{equation}{section}

\newcommand{\Z}{\mathbb{Z}}
 
\renewcommand{\paragraph}[1]{\medskip\noindent {\bf {#1}}}

\newcommand{\calA}{\ensuremath{\mathcal{A}}}
\newcommand{\calB}{\ensuremath{\mathcal{B}}}
\newcommand{\calC}{\ensuremath{\mathcal{C}}}

\newcommand{\calO}{\ensuremath{\mathcal{O}}}

\newcommand{\calS}{\ensuremath{\mathcal{S}}}

\newcommand{\zo}{\ensuremath{\{0,1\}}} 

  \theoremstyle{plain} 
  \newtheorem{theorem}{Theorem}[section]

  \theoremstyle{definition}

  \newtheorem{definition}[theorem]{Definition}  

\newtheorem{construction}[theorem]{Construction}

\newcommand{\esm}[1]{\ensuremath{#1}}
\newcommand{\ms}[1]{\esm{\mathsf{#1}}}

\newcommand{\poly}{\ms{poly}}
\newcommand{\negl}{\ms{negl}}

\newcommand{\zostar}{\zo^*}

\newcommand{\etal}{et~al.\xspace}

\newcommand{\rgets}{\mathrel{\mathpalette\rgetscmd\relax}}
\newcommand{\rgetscmd}{\ooalign{$\leftarrow$\cr
    \hidewidth\raisebox{1.2\height}{\scalebox{0.5}{\ \rm R}}\hidewidth\cr}}
\newcommand{\getsr}{\rgets}

\begin{document}

\title{Fast Privacy-Preserving Punch Cards}

\author{\IEEEauthorblockN{Saba Eskandarian}
\IEEEauthorblockA{Stanford University\\
\texttt{saba@cs.stanford.edu}}}


\ignore{
\IEEEoverridecommandlockouts
\makeatletter\def\@IEEEpubidpullup{6.5\baselineskip}\makeatother
\IEEEpubid{\parbox{\columnwidth}{
    Network and Distributed Systems Security (NDSS) Symposium 2020\\
    23-26 February 2020, San Diego, CA, USA\\
    ISBN 1-891562-61-4\\
    https://dx.doi.org/10.14722/ndss.2020.23xxx\\
    www.ndss-symposium.org
}
\hspace{\columnsep}\makebox[\columnwidth]{}}
}

\maketitle

\begin{abstract}
Loyalty programs in the form of punch cards that can be redeemed for benefits have long been a ubiquitous element of the consumer landscape. However, their increasingly popular digital equivalents, while providing more convenience and better bookkeeping, pose a considerable privacy risk. 
This paper introduces a privacy-preserving punch card protocol that allows firms to digitize their loyalty programs without forcing customers to submit to corporate surveillance. We also present a number of extensions that allow our scheme to provide other privacy-preserving customer loyalty features. 

Compared to the best prior work, we achieve a $14\times$ reduction in the computation and a $11\times$ reduction in the communication required to perform a ``hole punch,'' a $55\times$ reduction in the communication required to redeem a punch card, and a $128\times$ reduction in the computation time required to redeem a card. 
Much of our performance improvement can be attributed to removing the reliance on pairings or range proofs present in prior work, which has only addressed this problem in the context of more general loyalty systems. By tailoring our scheme to punch cards and related loyalty systems, we demonstrate that we can reduce communication and computation costs by orders of magnitude. 
\end{abstract}


\section{Introduction}\label{intro}

Punch cards that can be redeemed for rewards after a number of purchases are a widely-used incentive for customer loyalty. Although these time-tested loyalty schemes remain popular, they are increasingly being replaced with digital equivalents that reside in mobile apps instead of physical wallets. The benefits of going digital for business owners include stronger defenses against counterfeit cards, a more convenient customer experience, and better bookkeeping around the popularity and efficacy of their loyalty programs~\cite{article1,article2}. 

Unfortunately, digital loyalty programs also introduce myriad new opportunities for customers' privacy to be violated~\cite{article2,article3}, e.g., by linking customer behavior across transactions.
This kind of tracking can be conducted by the business itself, a third-party loyalty service, or a malicious actor who gains access in a data breach. 
Thus any firm who wants to protect customer privacy should attempt to ensure that its digital loyalty program does not collect unnecessary data. But is it possible to digitize the traditional punch card without damaging customer privacy?

\rev{One solution to this problem is to rely on classical techniques such as blind signatures, anonymous credentials, or Ecash~\cite{Chaum82,Chaum85,CL01,CL04,CHL05,BCKL09}. As a simple example, consider a scheme where, to give a customer a hole punch, the server produces a blind signature on a secret chosen by the customer. After the customer has accrued enough signatures, it reveals all of them to the server, who keeps a database of used secrets to prevent double spending. While this kind of scheme works, it requires storage linear in the number of punches \newrev{per card} on both the client and server, and card redemption requires time linear in the number of punches as well. }

A recent line of work, beginning with the Black Box Accumulation (BBA) of Jager and Rupp~\cite{bba}, addresses this problem \rev{by building schemes which keep track of a customer's ``balance'' of punches \newrev{on a card} in constant space. Unfortunately,} although individual hole punches are unlinkable in the original BBA scheme, the processes of issuing and redeeming a punch card are not. This shortcoming is rectified in a series of follow-up works~\cite{bbaplus, uacs,uacsplus,bbw,p4tc}, all of which additionally extend the idea of black box accumulation to support a broader set of functionalities\rev{, usually via some combination of blind signatures and zero-knowledge proofs that previously signed values satisfy certain relationships, e.g., that a balance was correctly updated.} 

This work introduces new protocols specifically designed to support privacy-preserving digital punch cards. By focusing specifically on the requirements of punch cards and similar points-based loyalty programs, we are able to make both qualitative and quantitative improvements over prior work. Unlike the works listed above, our main protocol does not rely on pairings or range proofs, enabling significant performance improvements. Moreover, by stepping away from previous abstractions used for punch cards, we can handle punch card issuance \emph{non-interactively}, meaning that a customer can generate a new, unpunched card without any interaction with the server. As an ancillary benefit, this removes a potential denial of service opportunity in prior systems, where a customer could register many punch cards without actually needing to earn any punches. 

In terms of performance, our scheme reduces the client side computation required to generate a new punch card by $171\times$ compared to prior work (in addition to not requiring interaction with the server), reduces the total client and server computation times to perform a card punch by $14\times$, and reduces the time to redeem a card by $128\times$. Communication costs to punch and redeem a card are also reduced by $11\times$ and $55\times$, respectively. \newrev{We outperform the blind signature-based scheme sketched above even for punch cards that require a single punch before redemption.}

Our core protocol is quite simple. To generate a punch card, a client picks a random secret and hashes it to a point in an elliptic curve group using a hash function modeled as a random oracle~\cite{FS86,BR93}. To receive a hole punch, the client masks this group element and sends it to the server, who sends it back raised to a server-side secret value, along with a proof that this was done honestly. Finally, after several punches, the client redeems the card by sending the unmasked version along with the initial random secret to the server. The server checks that the group element submitted matches the hash of the random secret raised to the appropriate exponent. It also checks that the punch card being redeemed has not been redeemed before. Since the server is not involved in card issuance and only ever sees separately masked versions of the card, it cannot link a redeemed card to any past transaction. We prove, in the Algebraic Group Model (AGM)~\cite{FKL18}, that a malicious customer cannot successfully claim more rewards than it is entitled to redeem.

We also present a number of extensions to our main scheme that allow us to handle variations on the typical punch card. For example, we can handle special promotions where users get multiple punches, programs where purchases receive a fixed number of points instead of a single punch, and even private ticketing systems. Our most involved extension allows customers to merge the points on two punch cards without revealing anything to the server about the individual punch cards being merged. This extension uses pairings, but it still maintains the other advantages of our protocol and outperforms prior work, albeit by a smaller margin. 

Our schemes are implemented in Rust with an Android wrapper for testing on mobile devices, and all our code and raw performance data are open source at 
\url{https://github.com/SabaEskandarian/PunchCard}.

\section{Design Goals}\label{goals}

The goal of a digital punch card scheme is to allow customers to have a ``punch card'' that resides on their mobile device and is punched, e.g., via NFC, at the time a purchase is made, instead of a physical card that must be handed to a store employee and physically punched with a hole puncher. 

This section \rev{begins by motivating privacy-preserving punch cards and discussing different deployment scenarios. Then} we give security definitions and contrast the goals of our work with those of closely related works. 

\subsection{\rev{Motivation and Deployment Scenarios}}

\rev{Given that digital loyalty programs are so often used to collect data on individual customers, we begin by reviewing the potential benefits of a privacy-preserving punch card scheme.} \rev{Digital punch cards} provide additional security and convenience to both the customer and the store. 

 \rev{First, there is the convenience for customers, who can carry fewer cards when the punch card programs they participate in are digital and can enjoy a smoother checkout experience. Although this indirectly benefits businesses, there are also direct benefits to having a digital punch card program, even if it does not track individual users. Aggregate statistics such as the number of cards punched or redeemed in a given day can be time-consuming and error-prone statistics for employees to record, especially during busy business hours, but a digital app can make keeping this information effortless. Digital solutions also save the time employees would otherwise spend punching cards, speeding up lines. Moreover, many physical punch cards are punched with hole punches that can easily be bought at an office supply store or on specialty stores online, meaning that counterfeiting is trivial for any motivated adversary. Going digital removes this possibility. }

\rev{Our scheme constitutes only part of a privacy-preserving solution to mobile apps for businesses, and the app in which it is deployed must also be designed to protect privacy. This is true of any privacy-preserving cryptographic scheme deployed within a larger product. We envision this deployment could either be a feature in a larger app that a business already provides for its customers, or there could be a stand-alone app that handles the punch card program for multiple businesses. The latter approach has the benefit that auditing the one app can suffice for assuring the privacy of punch cards at several businesses as opposed to auditing each implementation separately. }

We note that it may be the case that a customer's payment method already identifies them, e.g., if they use a credit card instead of cash, but having a private loyalty system is still important despite this. In particular, there are far more regulations on what payment networks like Visa/MasterCard can do with a customer's card information than what an arbitrary loyalty app that collects user data can do\rev{, so there is additional privacy harm in a loyalty app getting more customer information.} Moreover, new components added to the customer experience should be designed in a privacy-first way, so that the loyalty program does not become a barrier to privacy in the transaction process later.

\subsection{Functionality Goals}

A punch card scheme consists of three components. First, a client running on a customer's phone should be able to create a new punch card. Next, the client and a server running a loyalty program can interact in order for the server to give the client a ``hole punch.'' Finally, a client can submit a completed punch card to the server for verification, and the server will accept valid punch cards that have not already been redeemed. The server keeps a database \ms{DB} of previously redeemed cards to make sure a client doesn't redeem the same card multiple times. After verifying a card, the server can give the client some out-of-band reward. In general, each of these steps can be a multi-round interactive protocol between the two parties. However, since all our protocols involve exactly one round, we present the syntax of a punch card scheme below as consisting of individual algorithms instead of interactive protocols. 

A \emph{punch card scheme} defined with respect to a security parameter $\lambda$ is defined as follows. 

\begin{itemize}
\item $\ms{ServerSetup}(1^\lambda)\rightarrow \ms{sk}, \ms{pk}, \ms{DB}$: On input a security parameter $\lambda$, the initial server setup produces server public and secret keys, as well as an empty database to record previously redeemed punch cards. 

\item $\ms{Issue}(1^\lambda)\rightarrow \ms{psk}, p$: On input a security parameter $\lambda$, the \ms{Issue} algorithm generates new punch card $p$ and a punch card secret \ms{psk}. 

\item $\ms{ServerPunch}(\ms{sk}, \ms{pk}, p) \rightarrow p', \pi$: On input the server keys and a punch card $p$, \ms{ServerPunch} outputs an updated punch card $p'$ and a proof $\pi$ that the punch card $p$ was updated correctly. 

\item $\ms{ClientPunch}(\ms{pk}, \ms{psk}, p, p', \pi) \rightarrow \ms{psk}', p'' \text{or} \bot$: Given the public key, a punch card secret \ms{psk}, the accompanying punch card $p$, a server-updated punch card value $p'$, and a proof $\pi$, \ms{ClientPunch} outputs an updated secret $\ms{psk}'$ and card $p''$ if the proof $\pi$ is accepted and $\bot$ otherwise. 

\item $\ms{ClientRedeem}(\ms{psk}, p) \rightarrow \ms{psk}', p'$: Given a punch card secret \ms{psk} and the corresponding punch card $p$, \ms{ClientRedeem} outputs an updated secret $\ms{psk}'$ and card $p'$ that are ready to be sent to the server for redemption. 

\item $\ms{ServerVerify}(\ms{sk}, \ms{pk}, \ms{DB}, \ms{psk}, p, n) \rightarrow 1/0, \ms{DB}'$: on input the server keys, redeemed card database, a punch card, the accompanying secret, and an integer $n\in\Z$ determining the required number of punches for redemption, \ms{ServerVerify} outputs a bit determining whether or not the punch card is accepted and an updated database $\ms{DB}'$. 
\end{itemize}

Correctness for a punch card scheme is defined in a straightforward way. An honestly generated punch card that has received $n$ punches should be accepted by an honest server. This should hold true even after many punch cards have been generated and redeemed. 

\begin{definition}[Correctness]
We say that a punch card scheme is \emph{correct} if for 
\begin{align*}
&\ms{sk}, \ms{pk}, \ms{DB}_0\gets\ms{ServerSetup}(1^\lambda)
\end{align*} and any $n\in\Z$, the following set of operations, repeated sequentially $N=\poly(\lambda)$ times, results in $b_j=1$ for all $j\in[N]$ with all but negligible probability in $\lambda$.
\begin{align*}
&(\ms{psk}_0, p_0)\gets\ms{Issue}(1^\lambda)\\
&\ms{for\ }i\in[n]:\\
&\ms{\ \ \ \ }p_i', \pi_i \gets \ms{ServerPunch}(\ms{sk}, \ms{pk}, p_i)\\
&\ms{\ \ \ \ }\ms{psk}_{i+1}, p_{i+1} \gets \ms{ClientPunch}(\ms{pk}, \ms{psk}_i, p_i, p_i', \pi_i)\\
&\ms{psk}, p \gets \ms{ClientRedeem}(\ms{psk}_n, p_n)\\
&b_j, \ms{DB}_{j+1}\gets\ms{ServerVerify}(\ms{sk}, \ms{pk}, \ms{DB}_j, \ms{psk}, p, n)\\
\end{align*}
\end{definition}

The functionality we desire from our punch cards is at a high level similar to that offered by black box accumulation (BBA)~\cite{bba}. Although we offer a similar functionality, we will do so with stronger security guarantees and significantly improved performance. On the other hand, other follow-up works to BBA~\cite{bbaplus,uacs,uacsplus,bbw} offer additional features that might be useful in other kinds of loyalty programs, such as reducing balances and partially spending accrued rewards. These features enable other applications, but, as described in Section~\ref{intro}, they render the solutions less effective for the original punch card problem. Bobolz~\etal~\cite{uacsplus} introduce the possibility of recovering from a partially completed spend that gets interrupted mid-protocol, e.g., due to a communication or hardware fault. Our scheme avoids the potential for this problem entirely because redemption only requires a single message from the client to the server. 

One way in which our setting differs fundamentally from prior work is the way in which we prevent a punch card from being redeemed more than once. In our setting, the server has access to a database of all previously redeemed cards when deciding whether or not to accept a new punch card submitted for verification. Prior works consider an \emph{offline double spending} scenario where the server may not have access to such a database but must be able to identify clients who have double spent punch cards after the fact. We do not pursue this goal for three reasons, listed in order of increasing importance below.

\begin{enumerate}
\item Not necessary: point-of-sale terminals often require an internet connection to work, so synchronizing spent punch cards between different locations of a firm with multiple branches can happen online with less performance cost than an offline verification approach.

\item Prohibitively expensive: the performance cost of checking whether a punch card was double spent in prior work is often prohibitive, requiring an exponentiation for each previously redeemed punch card. This would be about 8 orders of magnitude slower than the hash table lookup required in our setting (as measured on our evaluation setup). 

\item Requires real-world identity: identifying the human user who double spent a punch card in a way that the person can be penalized requires some notion of real-world identity tied to the punch card client. This means that any loyalty system providing such a feature would require a user's real-world identity in order to operate. This violates our original goal of making a punch card loyalty program digital with no damage to user privacy. 
\end{enumerate}

\subsection{Security Goals}

\begin{figure*}
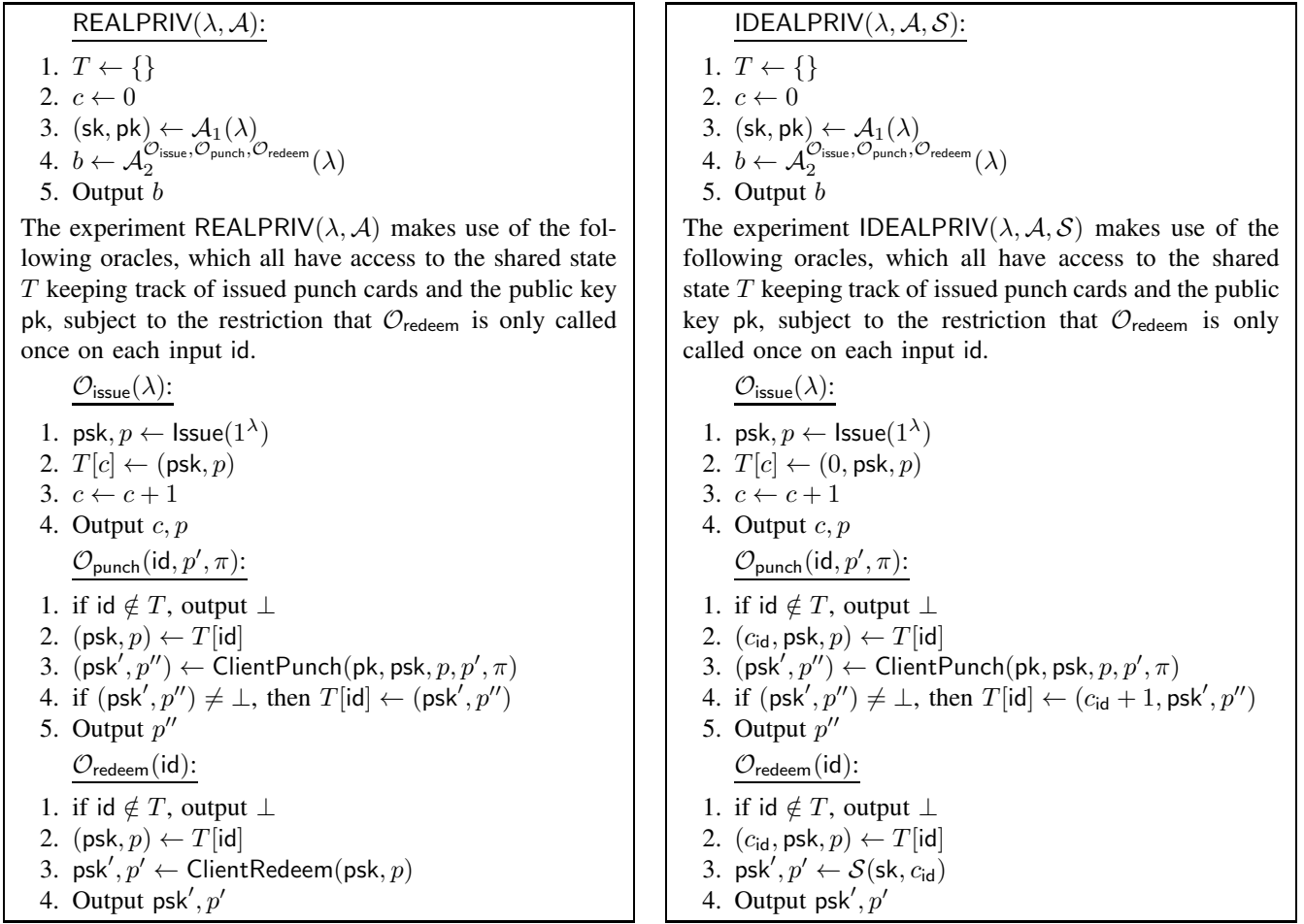
\centering
\fbox{
\begin{minipage}{.45\linewidth}
\begin{enumerate}
\item [] \underline{$\ms{REALPRIV}(\lambda, \calA)$:}\vspace{.5em}
\item [1.] $T\gets\{\}$
\item [2.] $c\gets 0$
\item [3.] $(\ms{sk}, \ms{pk})\gets \calA_1(\lambda)$
\item [4.] $b\gets \calA_2^{\calO_{\ms{issue}}, \calO_{\ms{punch}}, \calO_{\ms{redeem}}}(\lambda)$
\item [5.] Output $b$
\end{enumerate}

The experiment $\ms{REALPRIV}(\lambda, \calA)$ makes use of the following oracles, which all have access to the shared state $T$ keeping track of issued punch cards and the public key \ms{pk}, subject to the restriction that $\calO_\ms{redeem}$ is only called once on each input~$\ms{id}$. 

\begin{itemize}
\item [] \underline{$\calO_\ms{issue}(\lambda)$:}\vspace{.5em}
\item [1.] $\ms{psk}, p\gets\ms{Issue}(1^\lambda)$
\item [2.] $T[c]\gets(\ms{psk}, p)$
\item [3.] $c\gets c+1$
\item [4.] Output $c, p$
\end{itemize}

\begin{itemize}
\item [] \underline{$\calO_\ms{punch}(\ms{id}, p', \pi)$:}\vspace{.5em}
\item [1.] if $\ms{id}\notin T$, output $\bot$
\item [2.] $(\ms{psk}, p)\gets T[\ms{id}]$
\item [3.] $(\ms{psk}', p'')\gets\ms{ClientPunch}(\ms{pk}, \ms{psk}, p, p', \pi)$
\item [4.] if $(\ms{psk}', p'')\neq\bot$, then $T[\ms{id}] \gets (\ms{psk}', p'')$
\item [5.] Output $p''$
\end{itemize}

\begin{itemize}
\item [] \underline{$\calO_\ms{redeem}(\ms{id})$:}\vspace{.5em}
\item [1.] if $\ms{id}\notin T$, output $\bot$
\item [2.] $(\ms{psk}, p)\gets T[\ms{id}]$
\item [3.] $\ms{psk}', p'\gets\ms{ClientRedeem}(\ms{psk}, p)$
\item [4.] Output $\ms{psk}', p'$
\end{itemize}
\end{minipage}
}
\hspace{.5em}
\fbox{
\begin{minipage}{.45\linewidth}
\begin{enumerate}
\item [] \underline{$\ms{IDEALPRIV}(\lambda, \calA, \calS)$:}\vspace{.5em}
\item [1.] $T\gets\{\}$
\item [2.] $c\gets 0$
\item [3.] $(\ms{sk}, \ms{pk})\gets \calA_1(\lambda)$
\item [4.] $b\gets \calA_2^{\calO_{\ms{issue}}, \calO_{\ms{punch}}, \calO_{\ms{redeem}}}(\lambda)$
\item [5.] Output $b$
\end{enumerate}

The experiment $\ms{IDEALPRIV}(\lambda, \calA, \calS)$ makes use of the following oracles, which all have access to the shared state $T$ keeping track of issued punch cards and the public key \ms{pk}, subject to the restriction that $\calO_\ms{redeem}$ is only called once on each input~$\ms{id}$.

\begin{itemize}
\item [] \underline{$\calO_\ms{issue}(\lambda)$:}\vspace{.5em}
\item [1.] $\ms{psk}, p\gets\ms{Issue}(1^\lambda)$
\item [2.] $T[c]\gets(0, \ms{psk}, p)$
\item [3.] $c\gets c+1$
\item [4.] Output $c, p$
\end{itemize}

\begin{itemize}
\item [] \underline{$\calO_\ms{punch}(\ms{id}, p', \pi)$:}\vspace{.5em}
\item [1.] if $\ms{id}\notin T$, output $\bot$
\item [2.] $(c_\ms{id}, \ms{psk}, p)\gets T[\ms{id}]$
\item [3.] $(\ms{psk}', p'')\gets\ms{ClientPunch}(\ms{pk}, \ms{psk}, p, p', \pi)$
\item [4.] if $(\ms{psk}', p'')\neq\bot$, then $T[\ms{id}] \gets (c_\ms{id}+1, \ms{psk}', p'')$
\item [5.] Output $p''$
\end{itemize}

\begin{itemize}
\item [] \underline{$\calO_\ms{redeem}(\ms{id})$:}\vspace{.5em}
\item [1.] if $\ms{id}\notin T$, output $\bot$
\item [2.] $(c_\ms{id}, \ms{psk}, p)\gets T[\ms{id}]$
\item [3.] $\ms{psk}', p'\gets\calS(\ms{sk}, c_\ms{id})$
\item [4.] Output $\ms{psk}', p'$
\end{itemize}
\end{minipage}
}
\vspace{.5em}
\caption{\centering Real and ideal privacy experiments}\label{privexps}

\end{figure*}

\begin{figure}
\fbox{\begin{minipage}{.95\linewidth}
\begin{enumerate}
\item [] \underline{$\ms{SOUND}(\lambda, \calA)$:}\vspace{.5em}
\item [1.] $\ms{sk}, \ms{pk}, \ms{DB}\gets\ms{ServerSetup}(1^\lambda)$
\item [2.] $T \gets\{\}$
\item [3.] $c\gets0$
\item [4.] $c_\ms{punch}\gets 0$
\item [5.] $c_\ms{redeem}\gets 0$

\item [6.] $\calA^{\calO_{\ms{punch}}, \calO_{\ms{redeem}}, \calO_{\ms{honIssue}}, \calO_{\ms{honPunch}}, \calO_{\ms{honRedeem}}, \calO_{\ms{corrupt}}}(\lambda, \ms{pk})$
\item [7.] if $c_\ms{redeem} > c_\ms{punch}$, output 1. Otherwise, output 0. 
\end{enumerate}
\vspace{.5em}

The experiment $\ms{SOUND}(\lambda, \calA)$ makes use of the following oracles, which all have access to the shared state $T, c, c_\ms{punch}, c_\ms{redeem},\ms{sk},\ms{pk}, \ms{DB}$. \vspace{.5em}

\begin{itemize}
\item [] \underline{$\calO_\ms{punch}(p)$:}\vspace{.5em}
\item [1.] $p', \pi \gets \ms{ServerPunch}(\ms{sk}, \ms{pk}, p)$
\item [2.] $c_\ms{punch} \gets c_\ms{punch}+1$
\item [3.] Output $(p', \pi)$
\end{itemize}
\begin{itemize}
\item [] \underline{$\calO_\ms{redeem}(\ms{psk}, p, n)$:}\vspace{.5em}
\item [1.] $b, \ms{DB}'\gets\ms{ServerVerify}(\ms{sk}, \ms{pk}, \ms{DB}, \ms{psk}, p, n)$
\item [2.] if $b=1$:
\item [3.] \hspace{1em} $c_\ms{redeem}\gets c_\ms{redeem}+n$
\item [4.] \hspace{1em} $\ms{DB}\gets\ms{DB}'$
\item [5.] Output $b$
\end{itemize}
\begin{itemize}
\item [] \underline{$\calO_\ms{honIssue}()$:}\vspace{.5em}
\item [1.] $\ms{psk}, p\gets\ms{Issue}(1^\lambda)$
\item [2.] $T[c]\gets(0, \ms{psk}, p)$
\item [3.] $c\gets c+1$
\item [4.] Output $p$
\end{itemize}
\begin{itemize}
\item [] \underline{$\calO_\ms{honPunch}(\ms{id})$:}\vspace{.5em}
\item [1.] if $\ms{id}\notin T$, output $\bot$
\item [2.] $(c_\ms{id}, \ms{psk}, p)\gets T[\ms{id}]$
\item [3.] $p', \pi \gets \ms{ServerPunch}(\ms{sk}, \ms{pk}, p)$
\item [4.] $(\ms{psk}', p'')\gets\ms{ClientPunch}(\ms{pk}, \ms{psk}, p, p', \pi)$
\item [5.] if $(\ms{psk}', p'')\neq\bot$, then $T[\ms{id}] \gets (c_\ms{id}+1, \ms{psk}', p'')$
\item [6.] Output $p', \pi, p''$
\end{itemize}
\begin{itemize}
\item [] \underline{$\calO_\ms{honRedeem}(\ms{id})$:}\vspace{.5em}
\item [1.] if $\ms{id}\notin T$, output $\bot$
\item [2.] $(c_\ms{id}, \ms{psk}, p)\gets T[\ms{id}]$
\item [3.] $\ms{psk}', p'\gets\ms{ClientRedeem}(\ms{psk}, p)$
\item [4.] $b, \ms{DB}'\gets\ms{ServerVerify}(\ms{sk}, \ms{pk}, \ms{DB}, \ms{psk}', p', n)$
\item [5.] if $b=1$, then $\ms{DB}\gets\ms{DB}'$
\item [6.] Output $\ms{psk}', p', b$
\end{itemize}
\begin{itemize}
\item [] \underline{$\calO_\ms{corrupt}(\ms{id})$:}\vspace{.5em}
\item [1.] if $\ms{id}\notin T$, output $\bot$
\item [2.] $(c_\ms{id}, \ms{psk}, p)\gets T[\ms{id}]$
\item [3.] delete $\ms{id}$ from $T$
\item [4.] $c_\ms{punch}\gets c_\ms{punch}+c_\ms{id}$
\item [5.] Output $(c_\ms{id}, \ms{psk}, p)$
\end{itemize}

\end{minipage}}
\vspace{.5em}
\caption{\centering Soundness experiment}\label{expsound}
\end{figure}

At a high level, a punch card scheme must provide two kinds of security guarantees. First, it must protect client privacy such that the server cannot link messages sent by the same client. Second, it must be sound in that no client can redeem more rewards than it has honestly accrued through valid hole punches authorized by the server. 

We formally define privacy using a simulation-based definition. This means that in order for privacy to be satisfied, there must exist a \emph{simulator} algorithm that can generate the view of the punch card server without access to client-side secrets. Informally, if the server can't distinguish between the output of the simulator and a real client, then it surely can't learn anything from interacting with a real client because it could have received the same information by running the simulator on its own. Thus nothing is leaked that links one interaction with the server with other interactions. \rev{For simplicity, we do not prove security in the UC framework~\cite{UC,UCJournal}, but we do not anticipate any significant obstacles to composability in our protocols. }

Defining privacy for hole punches is rather straightforward. We simply require that the punch cards output by successful calls to $\ms{Issue}$ and $\ms{ClientPunch}$ can always be simulated regardless of the functions' inputs. We need not consider cases where the output is $\bot$ because in this case the client in a deployment of the punch card scheme would not return to the server with the updated punch card, so there is no possibility of later interactions being linked. 

Handling redemption is a little more involved because the view of the server during redemption may depend on the number of earlier punches on the same card. For this reason, our privacy definition for redemption defines real and ideal privacy experiments, both of which begin with the challenger initializing an empty table $T$ mapping unique integer identifiers to punch cards and a counter $c\gets 0$ that is incremented each time a new punch card is issued. The adversary is allowed to pick server secret and public keys $(\ms{sk}, \ms{pk})$, and then it is allowed to interact with oracles $\calO_\ms{issue}$, $\calO_\ms{punch}$, and $\calO_\ms{redeem}$ which play the role of the client in the punch card scheme. In the real privacy experiment, these oracles act as wrappers around the \ms{Issue}, \ms{ClientPunch}, and \ms{ClientRedeem} functions, simply calling the functions on the requested punch card (identified by an \ms{id} number chosen at issuance) and performing bookkeeping when punch cards are issued, updated, or redeemed. The ideal privacy experiment has identical $\calO_\ms{issue}$ and $\calO_\ms{punch}$ oracles, except it additionally keeps track of how many punches each card has received. The ideal $\calO_\ms{redeem}$ oracle, instead of calling $\ms{ClientRedeem}$, calls the simulator algorithm $\calS$, which is given no information about the card being redeemed except the number of punches it has previously received. At the end of each experiment, the adversary outputs a distinguishing bit~$b$. This definition implies that no adversary can distinguish between redeemed punch cards that have the same number of punches on them. Taken together with the simulators for punches, this implies that the server can simulate all interactions it has with clients.

\begin{definition}[Privacy]\label{def:privacy}
Let $\Pi$ be a punch card scheme. We say that $\Pi$ has \emph{privacy} if for security parameter $\lambda$, and for every adversary $\calA=(\calA_\ms{punch}, \calA_\ms{redeem})$, with $\calA_\ms{redeem}$ further subdivided into $\calA_1, \calA_2$, there exists a negliglible function $\ms{negl}(\cdot)$ and simulator algorithms $\calS_\ms{punch}$ and $\calS_\ms{redeem}$ such that the following conditions hold.
\begin{itemize}
\item For any punch card $p\neq\bot$ output by $\ms{Issue}(1^\lambda)$,
\begin{multline*}
\Big|\ms{Pr}\big[\calA_\ms{punch}(p)\big] - \ms{Pr}\big[\calA_\ms{punch}(\calS_\ms{punch}(1^\lambda))\big] \Big| < \ms{negl}(\lambda).
\end{multline*}

\item For any punch card $p''\neq\bot$ output by $\ms{ClientPunch}(\ms{pk}, \allowbreak \ms{psk}, \allowbreak p, \allowbreak p', \pi)$ for any choice of $(\ms{pk}, \ms{psk}, p, p', \pi)$,
\begin{multline*}
\Big|\ms{Pr}\big[\calA_\ms{punch}(p'')\big] - \ms{Pr}\big[\calA_\ms{punch}(\calS_\ms{punch}(1^\lambda))\big] \Big| < \ms{negl}(\lambda).
\end{multline*}

\item It holds that
\begin{multline*}
\Big|\ms{Pr\big[REALPRIV}(\lambda, \calA_\ms{redeem})=1\ms{\big]} \\- \ms{Pr\big[IDEALPRIV}(\lambda, \calA_\ms{redeem}, \calS_\ms{redeem})=1\ms{\big]}\Big|
<\ms{negl}(\lambda),
\end{multline*}
where the experiments \ms{REALPRIV} and \ms{IDEALPRIV} are defined as in Figure~\ref{privexps}. 
\end{itemize}
\end{definition}

The starting point for our soundness definition resembles that of BBA~\cite{bba}, which requires that a malicious client can only redeem as many punches as it has accrued. Aside from modifying the syntax of the definition to match our own, we have also modified it to allow the adversary to interleave hole punches and redemptions instead of requiring that all redemptions occur at the end of the protocol. More importantly, our definition also requires that an adversary cannot steal punches from honest users. BBA only requires that the total number of punches redeemed does not exceed the total number of punches given, which does not rule out the possibility of a malicious user stealing punches from honest users. 
\rev{Our soundness definition also models \emph{adaptive} corruptions, where the adversary can corrupt a previously honest user at any time during the security experiment. }
The definition gives the adversary the ability to get hole punches and redeem cards, but it also allows the adversary to ask for honest users to generate, punch, and redeem cards, with the adversary eavesdropping on all communications and having the ability to corrupt honest users at any point, acquiring their punches and secrets.

\begin{definition}[Soundness]\label{def:soundness}
Let $\Pi$ be a punch card scheme. Then for a security parameter $\lambda$ and adversary $\calA$, we define the soundness experiment $\ms{SOUND}(\lambda, \calA)$ in Figure~\ref{expsound}.
We say that a punch card scheme $\Pi$ satisfies \emph{soundness} if there exists a negligible function $\ms{negl}(\cdot)$ such that for any efficient adversary $\calA$, we have
$$\ms{Pr[SOUND}(\lambda, \calA)=1\ms{]}<\ms{negl}(\lambda).$$
\end{definition}

As in BBA, this definition does not capture whether or not a client can transfer value from one punch card to another or merge separate, partially filled punch cards to redeem a single, larger card. In fact, it is not entirely clear if this kind of card merging is a malicious behavior to be avoided or a beneficial feature to be desired. This kind of merging appears to be difficult to do in our main construction, but we show how to extend our scheme to allow a limited degree of merging in Section~\ref{merge}.

\section{Privacy-Preserving Punch Cards}\label{punchcard}

This section describes our main punch card scheme. In addition to its quantitative improvements over prior work, which we measure in Section~\ref{eval}, our scheme has a number of other desirable properties: 

\begin{itemize}
\item Whereas most prior works make use of pairings, either because they rely on Groth-Sahai proofs~\cite{GS08} or Pointcheval-Saunders signatures~\cite{PS16}, our punch card scheme does not require pairings.

\item We require no communication at all to issue a new punch card -- a client can do this on its own without server involvement. This removes a potential denial of service opportunity present in prior work, where a client could initiate a number of punch cards without making any purchases, thereby making the server incur unnecessary storage and computation at no cost to the malicious client. 

\item Our redemption process involves a client sending a single message to the server, so there is no potential for the process to be interrupted mid-protocol and no need for a recovery process of the form proposed by Bobolz~\etal~\cite{uacsplus}. 
\end{itemize}

\subsection{Main Construction}
\paragraph{A basic scheme}. We will begin with a bare-bones version of our scheme that provides neither privacy nor soundness. From this starting point, we will gradually build up to our actual scheme. Throughout, we will work in a group $G$ of prime order $q$. 

To set up the initial scheme, the server chooses a secret $\ms{sk}\in\Z_q$, and a client chooses a group element $p_0\getsr G$ to represent the punch card. To receive a hole punch, the client sends $p_i$ to the server, who returns $p_{i+1}\gets p_i^\ms{sk}$. To redeem a card after $n$ punches, the client submits $p_0$ and $p_{n}$ to the server, who accepts if $p_{n}=p_0^{\ms{sk}^n}$ and $p_0$ has not been previously used in a redeemed card. \rev{A nice feature of this approach, which we will keep in our final construction, is that punching a card requires no modification of server-side state.}

\paragraph{Adding privacy}. The scheme above clearly provides no privacy because the server can link the different times it sees a punch card. We can make punches made on the same card unlinkable by only sending the server \emph{masked} versions of the punch card, in a way reminiscent of standard oblivious PRF constructions~\cite{NR04,FIPR05}. The punch card is always masked with a fresh value $m\getsr \Z_q$ before being sent to the server, so the server only sees $p'\gets p^m$, not $p$ itself. The mask $m$ is removed (via exponentiation by $1/m$) before the next mask is applied. This means that the server sees a different random group element each time it punches a card. Moreover, an honest server only sees a random group element $p\in G$ and $p^{\ms{sk}^n}$ at redemption time. 

Unfortunately, this does not actually suffice to provide privacy against an actively malicious server. Consider a malicious server who always follows the scheme above, but during one hole punch (for a client it later wishes to re-identify) it uses a different secret $\ms{sk}'\getsr\Z_q$ so that $\ms{sk}\neq\ms{sk}'$ except with negligible probability. Then when an unsuspecting client attempts to redeem its punch card, instead of submitting $p, p^{\ms{sk}^n}$, it really submits $p, p^{\ms{sk}^{n-1}\ms{sk}'}$, allowing the server to identify it. 

We can handle the attack above by having the server give a zero knowledge proof of knowledge that it has honestly punched a card. To facilitate this, we require the server setup to also output a public key $\ms{pk}\gets g^\ms{sk}$, for some publicly known generator $g\in G$. Then the server can prove at punching time that it is returning a punch card $p'$ such that $p'=p^\ms{sk}$, i.e., that $p, \ms{pk}, p'$ form a DDH tuple~\cite{DH76}. This can be proven efficiently with a generic Chaum-Pedersen proof~\cite{CP92} made non-interactive in the random oracle model~\cite{FS86,BR93}. The server generates the proof $\pi$ and sends it to the client along with the punched card $p_{i+1}$. The client rejects the updated card if the proof does not verify. We denote proofs using the notation of Camenisch and Stadler~\cite{CS97}, where $ZKPK\{(\ms{sk}), \ms{pk}=g^\ms{sk}, p'=p^\ms{sk}\}$ represents the Chaum-Pedersen proof, and require the standard zero knowledge and existential soundness properties~\cite{cryptobook}. 

\paragraph{Adding soundness}. The two modifications above ensure that the scheme provides privacy, and they even protect against theft of honest clients' secrets by eavesdropping malicious clients because all the messages sent by honest clients to the server are uniformly random until the point where the punch card is redeemed (and then it's too late for the card to be stolen). Unfortunately, the scheme still fails to provide soundness, as a malicious client can redeem more points than it has received punches. Consider a client who at first honestly follows the protocol and redeems a punch card by submitting $p_0, p_n$. Next, it submits a masked $p_n$ for another punch and gets back $p_{n+1}$. Finally, it submits $p_1,p_{n+1}$ as another valid punch card. According to the scheme described thus far, the server would accept this punch card redemption, meaning that the malicious client can redeem $2n$ punches even though it only received $n+1$ punches. 

The attack above works because the client can choose any group element it wants as $p_0$. We modify our scheme to provide soundness by forcing clients to generate $p_0$ as the output of a hash function modeled as a random oracle $H:\zo^\lambda\to G$. In particular, instead of choosing a random $p_0$, the client chooses a random $u\gets\zo^\lambda$ and sets $p_0\gets H(u)$. When redeeming a punch card, instead of sending $p_0, p_n$, the client sends $u, p_n$, and the server checks that $p_n = H(u)^{\ms{sk}^n}$. Since the hash function is modeled as a random function, a malicious client cannot find the preimage of a group element under $H$, eliminating the attack. 

With this defense, our scheme now provides both privacy and soundness. 
We formalize our construction as follows.

\begin{construction}[Punch Card Scheme]
Let $G$ be a group of prime order $q$ with generator $g\in G$, and let $H$ be a hash function $H:\zostar\to G$, modeled as a random oracle.

We construct our punch card scheme as follows:

\begin{itemize}
\item $\ms{ServerSetup}(1^\lambda)\rightarrow \ms{sk}, \ms{pk}, \ms{DB}$: Select random $\ms{sk}\getsr \Z_q$ and set $\ms{pk}\gets g^\ms{sk}\in G$. Initialize \ms{DB} as an empty hash table, and return \ms{sk}, \ms{pk}, and \ms{DB}. 

\item $\ms{Issue}(1^\lambda)\rightarrow \ms{psk}, p$: First, select a random secret $u\getsr\zo^\lambda$ and a random masking value $m\getsr \Z_q$.  Then compute $p\gets H(u)^m\in G$. Let $\ms{psk}\gets(u,m)$. Return $\ms{psk}, p$.

\item $\ms{ServerPunch}(\ms{sk}, \ms{pk}, p) \rightarrow p', \pi$: Compute $p'\gets p^{\ms{sk}}$ as well as the proof of knowledge $\pi\gets ZKPK\{(\ms{sk}), \ms{pk}=g^\ms{sk}, p'=p^\ms{sk}\}$. Output $p',\pi$. 

\item $\ms{ClientPunch}(\ms{pk}, \ms{psk}, p, p', \pi) \rightarrow \ms{psk}', p'' \text{ or } \bot$: First, verify the proof $\pi$. If verification fails, output $\bot$. Otherwise, begin by interpreting $\ms{psk}$ as $(u,m)$. Then sample a new random masking value $m'\getsr\Z_q$ and compute $p''\gets (p')^{m'/m}$. Set $\ms{psk}'\gets(u,m')$, and output $\ms{psk}', p''$. 

\item $\ms{ClientRedeem}(\ms{psk}, p) \rightarrow \ms{psk}', p'$: Begin by interpreting $\ms{psk}$ as $(u, m)$ with $u\in\zo^\lambda$ and $m\in Z_q$. Then compute $p'\gets p^{1/m}\in G$. Return $u$ (as $\ms{psk}'$) and $p'$. 

\item $\ms{ServerVerify}(\ms{sk}, \ms{pk}, \ms{DB}, \ms{psk}, p, n) \rightarrow 1/0, \ms{DB}'$: Check whether $p=H(\ms{psk})^{\ms{sk}^n}$ and whether $\ms{psk}\in\ms{DB}$. If the first check returns true and the second returns false, insert $\ms{psk}$ into \ms{DB} and return $1, \ms{DB}$. Otherwise, return $0, \ms{DB}$.  
\end{itemize}
\end{construction}

Observe that the asymptotic complexity of almost every operation in our punch card scheme depends only on the security parameter $\lambda$, with two exceptions. The first exception is that operations on $\ms{DB}$ have amortized time complexity $O(\lambda)$, but in the worst case a read/write to $\ms{DB}$ could depend on the number of previously redeemed punch cards. The other exception is the exponentiation $\ms{sk}^n$ performed in \ms{ServerVerify}, where $O(\log n)$ group operations are required. However, since the same $n$ is often used for every punch card in practice, the server could precompute $\ms{sk}^n$ to remove the logarithmic dependence on $n$. 

\subsection{Security}

We now discuss the security of our constructions. We begin by proving the privacy of our punch card scheme. 

\begin{theorem}\label{thm:privacy}
Assuming the existential soundness of the Chaum-Pedersen proof system, our punch card scheme has privacy (Definition~\ref{def:privacy}) in the random oracle model. 
\end{theorem}

\begin{proof}
We begin by describing the simulator $\calS = (\calS_\ms{punch}, \calS_\ms{redeem})$.

\begin{itemize}
\item $\calS_\ms{punch}(1^\lambda)\rightarrow p$: This simulator samples and outputs a random group element $p\getsr G$.

\item $\calS_\ms{redeem}(\ms{sk}, c_\ms{id})\rightarrow \ms{psk}', p'$: This simulator samples a random string $\ms{psk}'\getsr\zo^\lambda$ and computes $p'\gets H(\ms{psk}')^{\ms{sk}^{c_\ms{id}}}$. It outputs $\ms{psk}', p'$. 
\end{itemize}

The first two conditions of privacy are clearly met by $\calS_\ms{punch}$ because the punch card output by a successful call to $\ms{Issue}$ or $\ms{ClientPunch}$ is always a group element raised to a randomly chosen mask $m$ (or $m'$), which will be distributed identically to a randomly chosen group element $p\getsr G$. 

Next, we show through a short series of hybrids that $\ms{REALPRIV}(\lambda, \calA) \allowbreak\approx_c\allowbreak\ms{IDEALPRIV}(\lambda,\calA, \calS_\ms{redeem})$ for our punch card scheme. 

\begin{itemize}
\item[\textbf{H0}:] This hybrid is the real privacy experiment $\ms{REALPRIV}(\lambda, \calA)$.
\item[\textbf{H1}:] In this hybrid, we add an abort condition to the execution of the experiment. The experiment aborts and outputs 0 if $\ms{ClientPunch}$ outputs $p''\neq\bot$ (i.e., it accepts the proof $\pi$) but it is not the case that $\ms{pk}=g^\ms{sk}\wedge p'=p^\ms{sk}$. 

This hybrid is indistinguishable from $H0$ by the soundness of the Chaum-Pedersen proof system. In particular, an adversary $\calA$ who can distinguish between $H0$ and $H1$ can be used by an algorithm $\calB$ to break the soundness of the proof system as follows. $\calB$ plays the role of the adversary in the soundness game for the Chaum-Pedersen proof, and plays the role of the challenger to $\calA$ in either $H0$ or $H1$ with probability $1/2$ each. Whenever $\calA$ causes experiment $H1$ to abort due to the check introduced in this hybrid, $\calB$ submits the proof $\pi$ and the statement $\ms{pk}=g^\ms{sk}\wedge p'=p^\ms{sk}$ to the soundness challenger. Otherwise, $\calB$ outputs $\bot$. 

The algorithm $\calB$ described above breaks the soundness of the Chaum-Pedersen proof with the same advantage that $\calA$ distinguishes between $H0$ and $H1$. To see why, observe that the only difference in the view of $\calA$ between $H0$ and $H1$ occurs when $H1$ aborts. Thus $\calA$ must cause the experiment to abort with probability at least equal to its distinguishing advantage between $H0$ and $H1$. But whenever $H1$ aborts, $\calB$ has a statement and proof that violate the soundness of the Chaum-Pedersen proof, so it wins the soundness game with the same advantage. This argument must be repeated for each proof given to the adversary. 

\item[\textbf{H2}:] In this hybrid, the challenger switches to record-keeping in the table $T$ in the way \ms{IDEALPRIV} does and replaces calls to \ms{ClientRedeem} with calls to $\calS_\ms{redeem}$. 

This hybrid is indistinguishable from $H1$ because the distribution of the adversary $\calA$'s view is identical in the two hybrids. The outputs of $\calO_\ms{issue}$ and $\calO_\ms{punch}$ are unchanged by the bookkeeping change between the two hybrids, so we need only consider $\calO_\ms{redeem}$.

\begin{itemize}
\item $\calO_\ms{redeem}(\ms{id})\rightarrow \ms{psk}', p'$: In $H1$, this oracle returns the secret $u$ used to generate the punch card stored at $T[\ms{id}]$ as well as the value of that punch card $p$ after removing the last mask $m$ to get $p'\gets p^{1/m}$. The value of $u$ is distributed uniformly at random in $\zo^\lambda$. The value of $p'$ is equal to $H(u)$ raised to the server secret $\ms{sk}$ as many times as there was a successful call to $\calO_\ms{punch}(\ms{id}, \cdot, \cdot)$ -- that is, a call whose output was not $\bot$. This is the case because in each such call, the punch card value stored in $T$ is raised to $\ms{sk}$ and its mask is replaced with a new one. The final unmasking operation $p'\gets p^{1/m}$ results in a punch card value $p'=H(u)^{\ms{sk}^n}$, where $n$ is the number of successful calls to $\calO_\ms{punch}(\ms{id}, \cdot, \cdot)$. 

In $H2$, $u$ clearly has the same distribution as in $H1$ because in $\calS_\ms{redeem}$ it is sampled directly as $u\getsr \zo^\lambda$. The value $p'$ also has the same distribution as in $H1$ because the table $T$ keeps count of the number $c_\ms{id}$ of successful calls to $\calO_\ms{punch}(\ms{id}, \cdot, \cdot)$, so $\calS_\ms{redeem}$ can compute $p'\getsr H(u)^{\ms{sk}^{c_\ms{id}}}$ directly. 
\end{itemize}

\item[\textbf{H3}:] This hybrid is identical to $H2$ except the abort condition introduced in $H1$ is removed. As was the case in $H1$, this hybrid is indistinguishable from the preceding hybrid by the soundness of the Chaum-Pedersen proof system. It also corresponds to the ideal privacy game $\ms{IDEALPRIV}(\lambda, \calA, \calS)$, completing the proof. 
\end{itemize}
\end{proof}

Having proven privacy, we now turn to soundness. We prove the soundness of our scheme in the algebraic group model (AGM)~\cite{FKL18}, where for every group element the adversary produces, it must also give a representation of that group element in terms of elements it has already seen. This is a strictly weaker model (in the sense that it puts fewer restrictions on the adversary) than the widely-used generic group model~\cite{Shoup97}, in which some of the prior works on privacy-preserving loyalty programs have been proven secure~\cite{uacs,uacsplus}. 
Our proof relies on the $q$-discrete log assumption, which assumes the computational hardness of winning the following game.

\begin{definition}[$q$-discrete log game]
The $q$-discrete log game for a group $G$ of prime order $p$ is played between a challenger $\calC$ and an adversary $\calA$. The challenger $\calC$ samples $x\getsr \Z_p$ and sends $g^x, g^{x^2},..., g^{x^q}$ to $\calA$. The adversary $\calA$ responds with a value $z\in\Z_p$, and the challenger outputs 1 iff $z=x$. 
\end{definition}

Depending on the concrete group in which the assumption is made, the $q$-discrete log game could be vulnerable to Brown-Gallant-Cheon attacks~\cite{BG04,Cheon06}, which reduce the security of the assumption by a factor of $\sqrt{q}$. Fortunately this attack only negligibly affects the security of the scheme, as $q$ is at most a polynomial in the security parameter $\lambda$. 

We now state and prove our soundness theorem. 
\begin{theorem}\label{thm:soundness}
Assuming the zero-knowledge property of the Chaum-Pedersen proof system and the $q$-discrete log assumption in $G$, our punch card scheme has soundness (Definition~\ref{def:soundness}) in the algebraic group model with random oracles. 
\end{theorem}

\begin{proof}
Since $q$ already refers to the order of the group $G$, we will refer to the $N$-discrete log assumption throughout this proof. The high-level idea of the proof is to program random oracle queries with re-randomizations of powers of $g^x$ given by the $N$-discrete log challenger. Then, whenever a punch card is given by the adversary, the algebraic adversary must also give a representation of the punch card $p$ in terms of group elements it has seen before. As such, the challenger can pick out the $g^{x^i}$ component and replace it with $g^{x^{i+1}}$ in its response. Then a punch card that is accepted before receiving $n$ punches must include a second representation of $g^{x^n}$, allowing us to solve for $x$.  

We now formalize the proof idea sketched above. Our proof proceeds through a series of hybrids. 

\begin{itemize}
\item[\textbf{H0}:] This hybrid is the soundness experiment $\ms{SOUND}(\lambda, \calA)$. 

\item[\textbf{H1}:] In this hybrid, we replace each proof $\pi$ output by $\calO_\ms{punch}$ and $\calO_\ms{honPunch}$ with a simulated proof. 

The the zero-knowledge property of the Chaum-Pedersen proof guarantees that the proofs can be simulated. Since hybrids $H0$ and $H1$ are identical save for the real proofs in $H0$ and the simulated proofs in $H1$, the output of an adversary $\calA$ who distinguishes between $H0$ and $H1$ can also be used to distinguish between a real and simulated proof with the same advantage. This argument must be repeated for each proof given to the adversary. 

\item[\textbf{H2}:] In this hybrid, we add an abort condition to the execution of the experiment. The experiment aborts and outputs 0 if the $\calO_\ms{redeem}(\ms{psk}, p,n)$ oracle ever outputs  $b=1$ when it receives a value of $\ms{psk}$ that the adversary has not previously queried from the random oracle $H$ or been given as the result of a query to $\calO_\ms{corrupt}$. 

This hybrid is indistinguishable from $H1$ because the probability of an adversary successfully triggering this abort condition is negligible in $\lambda$ and there are no other differences between $H1$ and $H2$. In order for the $\calO_\ms{redeem}$ oracle to output $b=1$, it must be the case that $\ms{ServerVerify}(\ms{sk}, \ms{pk}, \ms{DB},\ms{psk},p,n)$ outputs 1, which means that $p=H(\ms{psk})^{\ms{sk}^n}$. 

But since $H$ is modeled as a random function and $(\ms{psk}, H(\ms{psk}))$ is not known to the adversary, its output appears uniformly random in $G$. This is the case even for the punch cards of uncorrupted honest users because the messages the adversary sees from their interactions with the server are always uniformly random group elements (because of the masks, see privacy proof) and simulated proofs from the server. But then $H(\ms{psk})^{\ms{sk}^n}$ is also distributed uniformly at random in $G$, and the probability $\ms{Pr[}p=H(\ms{psk})^{\ms{sk}^n}\ms{]}\leq\negl(\lambda)$. 

\item[\textbf{H3}:] In this hybrid, we modify how the challenger computes the output of $\ms{ServerPunch}(\ms{sk},\ms{pk}, p)$ and of the random oracle $H$. Recall that since $\calA$ is an algebraic adversary, every group element it sends is accompanied by a representation in terms of the previous group elements it has seen: the generator $g$, returned punch cards $p'_1,...,p'_Q$ for the $Q$ queries it has made to the $\calO_\ms{punch}$ oracle, and random oracle outputs $H_1,...,H_{Q'}$ for the $Q'$ random oracle queries it has made. 

Let $g_i=g^{\ms{sk}^i}$ for $i\in\Z$. Whenever the adversary $\calA$ or oracle $\calO_\ms{honIssue}()$ makes a call to the oracle $H$ on a previously unqueried point $u$, the challenger samples $r\getsr\Z_q$ and sets $H(u)\gets g_1^r$. Since $r$ is distributed uniformly at random in $\Z_q$, so is~$H(u)$. 

Next, whenever $\calA$ makes a call to the oracles $\calO_\ms{punch}(p)$ or $\calO_\ms{honPunch}(\ms{id})$, instead of setting $p'\gets p^\ms{sk}$, the challenger looks at the algebraic representation of $p$ and replaces each occurrence of $g_i$ with $g_{i+1}$, including replacing $g$ with $g_1$. Since the only elements $\calA$ has seen are $g$, random oracle outputs, and the previous results of $\calO_\ms{punch}$, the challenger can keep track of which elements contain which $g_i$ as it sends them to $\calA$. 
The outputs of $\calO_\ms{punch}(p)$ in $H3$ are identical to the outputs in $H2$, because the process described here results in the same group element $p'$ that would be represented by~$p^\ms{sk}$. 

We now additionally have the challenger keep track of algebraic representations of the group elements it produces itself (e.g., in honest users' punch cards). This makes no changes to the adversary's view in the experiment. 

Since all the changes in $H3$ result in identically distributed outputs as in $H2$, the two hybrids are indistinguishable.
\end{itemize}

From $H3$, we can prove that any algebraic adversary $\calA$ who wins the soundness game can be used by an algorithm $\calB$, described below, to break the $N$-discrete log assumption in $G$. Algorithm $\calB$ plays the role of the adversary in the $N$-discrete log game while simultaneously playing the role of the challenger in $H3$. Algorithm $\calB$ simulates $H3$ exactly to $\calA$, except that it uses the $N$-discrete log challenge messages $g^x, g^{x^2}, ..., g^{x^N}$ as the values of $g_i$. That is, $g_i=g^{x^i}$. Moreover, it sets $\ms{pk}\gets g_1$ in the setup phase. Observe that the $g_i$ are distributed identically as in $H3$, so this is a perfect simulation of $H3$ with $x$ playing the role of $\ms{sk}$. The value of $N$ required in the assumption depends on the maximum number of sequential punches $\calA$ requests on the same group element.  

Now, if $\calA$ wins the soundness game, it means that $c_\ms{redeem}>c_\ms{punch}$. This, in turn, implies that there was some successful punch card redemption $\ms{ServerVerify}$ where the accepted value of $p$ had not been previously punched $n$ times, i.e., the representation of $p$ does not contain $g_{n+1}$. But since successful verification requires that $p=H(\ms{psk})^\ms{x^{n}}=(g_1^r)^\ms{x^n}=g_{n+1}^r$, and the algebraic adversary $\calA$ must give a representation of $p$, we now have two different representations of $g_{n+1}=g^{x^{n+1}}$, which together yield a degree-$n+1$ equation in $x$. This equation can be solved for $x$ using standard techniques~\cite{shoupbook}, allowing $\calB$ to recover $x$ and win the $N$-discrete log game. 
\end{proof}

\section{Merging Punch Cards}\label{merge}

Having described our main construction, we now consider another feature sometimes enjoyed by physical punch cards that we may want to reproduce digitally: merging partially-filled cards. Just as in real life, it is possible to ``merge'' two punch cards by redeeming them separately and taking into account the sum of the number of punches across the two cards. However, this process reveals the number of punches held by each card at redemption time, information that the customer may want to hide. We can hide the value of the two cards being merged by resorting to pairings.  

\begin{definition}[Pairings~\cite{cryptobook}]
Let $G_0,G_1,G_T$ be three cyclic groups of prime order $q$ where $g_0\in G_0$ and $g_1\in G_1$ are generators. A \emph{pairing} is an efficiently computable function $e:G_0 \times G_1 \to G_T$ satisfying the following properties:
\begin{itemize}
\item Bilinear: for all $u, u'\in G_0$ and $v,v'\in G_1$ we have 
$$e(u\cdot u',v)=e(u,v)\cdot e(u',v)$$
$$\text{ and } $$
$$e(u, v\cdot v')=e(u,v)\cdot e(u,v')$$
\item Non-degenerate: $g_T\gets e(g_0,g_1)$ is a generator of $G_T$. 
\end{itemize}

When $G_0=G_1$, we say that the pairing is a \emph{symmetric pairing}. We refer to $G_0$ and $G_1$ as the \emph{pairing groups} and refer to $G_T$ as the \emph{target group}. 
\end{definition}

Using a symmetric pairing, we can quite simply merge two punch cards without revealing the number of punches on each. Before redeeming punch cards $p_0$ and $p_1$ which have $i$ and $j$ punches, respectively, with $i+j=n$, the client computes $p\gets e(p_0,p_1)$. To redeem a merged card, the client sends the server the merged punch card $p$ along with $u_0$ and $u_1$, the secrets for the two punch cards merged into $p$. The server checks that $p=e(H(u_0)^{\ms{sk}^n}, H(u_1))$. The bilinear property of the pairing ensures that $e\big(p_0, p_1\big)=e\big(H(u_0)^{\ms{sk}^i}, H(u_1)^{\ms{sk}^j}\big)=e\big(H(u_0)^{\ms{sk}^n}, H(u_1)\big)$. We can even hide whether or not a redeemed punch card is merged by generating a fresh punch card before redemption and merging a complete card with it. 

The performance of symmetric pairings is far worse than that of asymmetric pairings, so we would like to have a scheme that works for asymmetric pairings as well. Unfortunately, directly converting the idea above to asymmetric pairings meets with some difficulties. Since each punch card must belong to either $G_0$ or $G_1$, we can only merge pairs of cards where $p_0\in G_0$ and $p_1\in G_1$. But this is a decision that must be made when a card is first issued, restricting punch cards to being merged with cards that belong to the other pairing group. 

We resolve this problem by splitting each punch card into two components, one in each pairing group. Each component behaves as a punch card in the original scheme. Generating a punch card is similar to the original scheme, but the secret $u\getsr\zo^\lambda$ is hashed by two different functions $H_0:\zo^\lambda\to G_0$ and $H_1:\zo^\lambda\to G_1$. Each hole punch repeats the punch protocol of the original scheme twice, once in $G_0$ and once in $G_1$. Redeeming a card requires merging the $G_0$ and $G_1$ components of the two cards with each other as above, and since the client has a version of each punch card in both groups, it can merge them as before. 

We formalize this sketch of a solution below. We replace the $\ms{ClientRedeem}$ algorithm from our punch card syntax with a new $\ms{ClientMergeRedeem}$ algorithm that merges two punch cards before redeeming them. 

\begin{construction}[Mergeable Punch Card Scheme]
Let $G_0,G_1,G_T$ be groups of prime order $q$ with generators $g_0\in G_0,g_1\in G_1$, and let $H_0,H_1$ be hash functions $H_0:\zostar\to G_0, H_1:\zostar\to G_1$, modeled as random oracles. 
We construct our punch card scheme as follows:

\begin{itemize}
\item $\ms{ServerSetup}(1^\lambda)\rightarrow \ms{sk}, \ms{pk}, \ms{DB}$: Select random $\ms{sk}\getsr \Z_q$ and set $\ms{pk}_0\gets g_0^\ms{sk}\in G_0, \ms{pk}_1\gets g_1^\ms{sk}\in G_1$. Initialize \ms{DB} as an empty hash table, and return \ms{sk}, $\ms{pk}=(\ms{pk}_0,\ms{pk}_1)$, and \ms{DB}. 

\item $\ms{Issue}(1^\lambda)\rightarrow \ms{psk}, p$: First, select a random secret $u\getsr\zo^\lambda$ and random masking values $m_0\getsr \Z_q, m_1\getsr\Z_q$.  Then compute $p_0\gets H_0(u)^{m_0}\in G_0, p_1\gets H_1(u)^{m_1}\in G_1$. Let $\ms{psk}\gets(u,m_0,m_1)$. Return $\ms{psk}, p=(p_0,p_1)$.

\item $\ms{ServerPunch}(\ms{sk}, \ms{pk}, p) \rightarrow p', \pi$: First, interpret $\ms{pk}$ as $(\ms{pk}_0,\ms{pk}_1)$ and $p$ as $(p_0,p_1)$. Compute $p_0'\gets p_0^{\ms{sk}}, p_1'\gets p_1^{\ms{sk}}$ as well as the proofs of knowledge $\pi_0\gets ZKPK\{(\ms{sk}), \ms{pk}_0=g_0^\ms{sk}, p_0'=p_0^\ms{sk}\}$ and $\pi_1\gets ZKPK\{(\ms{sk}), \ms{pk}_1=g_1^\ms{sk}, p_1'=p_1^\ms{sk}\}$. Output $p'=(p'_0,p'_1),\pi=(\pi_0, \pi_1)$. 

\item $\ms{ClientPunch}(\ms{pk}, \ms{psk}, p, p', \pi) \rightarrow \ms{psk}', p'' \text{or} \bot$: First, interpret $\ms{psk}$ as $(u, m_0, m_1)$, $p$ as $(p_0,p_1)$, $p'$ as $(p_0',p_1')$, and $\pi$ as $(\pi_0,\pi_1)$ . Next, verify the proofs $\pi_0$ and $\pi_1$. If either verification fails, output $\bot$. 
Then sample new random masking values $m_0'\getsr\Z_q, m_1'\getsr\Z_q$ and compute $p_0''\gets (p_0')^{m_0'/m_0}, p_1''\gets(p_1'^{m_1'/m_1})$. Finally, output $\ms{psk}'=(u, m_0', m_1'), p''=(p_0'',p_1'')$. 

\item $\ms{ClientMergeRedeem}(\ms{psk}, p, \ms{psk}', p') \rightarrow \ms{psk}'', p''$: 
Begin by interpreting $\ms{psk}$ as $(u, m_0, m_1)$, $\ms{psk}'$ as $(u', m_0', m_1')$, $p$ as $(p_0,p_1)$, and $p'$ as $(p_0',p_1')$. 
Then compute $p''\gets e(p_0^{1/m_0}, (p_1')^{1/m_1'})\in G_T$. Return $\ms{psk}''=(u,u')$ and $p''$. 

\item $\ms{ServerVerify}(\ms{sk}, \ms{pk}, \ms{DB}, \ms{psk}, p, n) \rightarrow 1/0, \ms{DB}'$: Begin by interpreting $\ms{psk}$ as $(u, u')$. Then perform the following checks:
\begin{enumerate}
\item $p = e(H_0(u)^{\ms{sk}^n}, H_1(u'))$

\item $u\in\ms{DB}$

\item $u'\in\ms{DB}$
\end{enumerate}

If the first check returns true and the other checks return false, insert $u$ and $u'$ into \ms{DB} and return $1, \ms{DB}$. Otherwise, return $0, \ms{DB}$.  
\end{itemize}
\end{construction}

Although not included in our formal construction, our scheme could be extended to allow more punches to occur on a merged card so long as the client indicates that it is a merged card being punched and the punch/proof occur over elements in $G_T$. Note that this scheme only allows for two punch cards to merged. Our general strategy for merging punch cards could be extended to more than two cards using multilinear maps~\cite{BS02,GGH13,CLT13}, but a construction that allows merging of more than two cards while only relying on efficient standard primitives would require new techniques. This is an interesting problem for future work to address. 

We now state and prove our security theorems for the mergeable punch card scheme. The only change required in the security games to account for the change from $\ms{ClientRedeem}$ to $\ms{ClientMergeRedeem}$ is that the redeem oracle in the privacy game takes in two $\ms{id}$s instead of just one and passes both corresponding punch cards to $\ms{ClientMergeRedeem}$. 

\begin{theorem}
Assuming the existential soundness of the Chaum-Pedersen proof system, our mergeable punch card scheme has privacy in the random oracle model. 
\end{theorem}

\begin{proof}[Proof (sketch)]
We begin by describing the simulator $\calS = (\calS_\ms{punch}, \calS_\ms{redeem})$.

\begin{itemize}
\item $\calS_\ms{punch}(1^\lambda)\rightarrow p$: This simulator samples and outputs two random group elements $p_0\getsr G_0$ and $p_1\getsr G_1$. 

\item $\calS_\ms{redeem}(\ms{sk}, c_\ms{id}, c_{\ms{id}'})\rightarrow \ms{psk}', p'$: This simulator samples two random strings $u\getsr\zo^\lambda, u'\getsr\zo^\lambda$ and computes $p'\gets e(H_0(u)^{\ms{sk}^{c_\ms{id}}}, H_1(u')^{\ms{sk}^{c_{\ms{id}'}}})$. It outputs $\ms{psk}\gets(u,u')$ and $p'$. 
\end{itemize}

The first two conditions of privacy are clearly met by $\calS_\ms{punch}$ because the punch card output by a successful call to $\ms{Issue}$ or $\ms{ClientPunch}$ is always two group elements raised to randomly chosen masks $m_0$ and $m_1$ (or $m_0'$ and $m_1'$), which will be distributed identically to randomly chosen group elements $p_0\getsr G_0, p_1\getsr G_1$. 

Next, we show through a short series of hybrids that $\ms{REALPRIV}(\lambda, \calA) \allowbreak\approx_c\allowbreak\ms{IDEALPRIV}(\lambda,\calA, \calS_\ms{redeem})$ for our mergeable punch card scheme. 
The rest of the proof of this theorem is very similar to that of Theorem~\ref{thm:privacy}. The main difference is that the soundness of the Chaum-Pedersen proof system needs to be invoked separately in each of $G_0$ and $G_1$. Thus we only sketch the steps of the hybrid argument below.

\begin{itemize}
\item[\textbf{H0}:] This hybrid is the real privacy experiment $\ms{REALPRIV}(\lambda, \calA)$.
\item[\textbf{H1}:] In this hybrid, we add an abort condition to the execution of the experiment. The experiment aborts and outputs 0 if \ms{ClientPunch} outputs $p''\neq\bot$ (i.e., it accepts the proofs $\pi_0$ and $\pi_1$) but it is not the case that $\ms{pk}_0=g_0^\ms{sk}\wedge p_0'=p_0^\ms{sk}$. 
\item[\textbf{H2}:] In this hybrid, we add an abort condition to the execution of the experiment. The experiment aborts and outputs 0 if \ms{ClientPunch} outputs $p''\neq\bot$ (i.e., it accepts the proofs $\pi_0$ and $\pi_1$) but it is not the case that $\ms{pk}_1=g_1^\ms{sk}\wedge p_1'=p_1^\ms{sk}$.

\item[\textbf{H3}:] In this hybrid, the challenger switches to record-keeping in the table $T$ in the way \ms{IDEALPRIV} does and replaces calls to \ms{ClientMergeRedeem} with calls to $\calS_\ms{redeem}$. 

\item[\textbf{H4}:] This hybrid is identical to $H3$ except the abort condition introduced in $H2$ is removed. 
\item[\textbf{H5}:] This hybrid is identical to $H4$ except the abort condition introduced in $H1$ is removed. It also corresponds to the ideal privacy game $\ms{IDEALPRIV}(\lambda, \calA, \calS)$, completing the proof. 

\end{itemize}

\end{proof}

Next, we prove the soundness of our scheme. Since our new scheme uses pairings, we use an \emph{asymmetric} $q$-discrete log assumption, which assumes the computational hardness of winning the following game.

\begin{definition}[asymmetric $q$-discrete log game]
The $q$-discrete log game for groups $G_0,G_1$ of prime order $p$ is played between a challenger $\calC$ and an adversary $\calA$. The challenger $\calC$ samples $x\getsr \Z_p$ and sends $g_0^x, g_0^{x^2},..., g_0^{x^q}, g_1^x, g_1^{x^2}, ..., g_1^{x^q}$ to $\calA$. The adversary $\calA$ responds with a value $z\in\Z_p$, and the challenger outputs 1 iff $z=x$. 
\end{definition}

\begin{theorem}
Assuming the zero-knowledge property of the Chaum-Pedersen proof system and the asymmetric $q$-discrete log assumption in $G_0$ and $G_1$, our mergeable punch card scheme has soundness (Definition~\ref{def:soundness}) in the algebraic group model with random oracles. 
\end{theorem}

\begin{proof}[Proof (sketch)]
Since $q$ already refers to the order of the groups $G_0,G_1,G_T$, we will refer to the asymmetric $N$-discrete log assumption throughout this proof. The majority of proof of this theorem is very similar to that of Theorem~\ref{thm:soundness}. The main difference in the hybrids is that several hybrids need to be repeated to account for each punch card being made up of two group elements instead of one. Thus we only sketch the steps of the hybrid argument below and focus on the last step of the argument. 

\begin{itemize}
\item[\textbf{H0}:] This hybrid is the soundness experiment $\ms{SOUND}(\lambda, \calA)$. 

\item[\textbf{H1}:] In this hybrid, we replace the proofs $\pi_0$ output by $\calO_\ms{punch}$ and $\calO_\ms{honPunch}$ with simulated proofs. 

\item[\textbf{H2}:] In this hybrid, we replace the proofs $\pi_1$ output by $\calO_\ms{punch}$ and $\calO_\ms{honPunch}$ with simulated proofs. 

\item[\textbf{H3}:] In this hybrid, we add an abort condition to the execution of the experiment. The experiment aborts and outputs 0 if the $\calO_\ms{redeem}(\ms{psk}, p,n)$ oracle ever outputs  $b=1$ when it receives a value of $u$ or $u'$ that the adversary has not previously queried from both random oracles $H_0$ and $H_1$ or been given as the result of a query to $\calO_\ms{corrupt}$. 

\item[\textbf{H4}:] In this hybrid, we modify how the challenger computes the output of $\ms{ServerPunch}(\ms{sk},\ms{pk}, p)$ and of the random oracle $H$. Let $g_{0,i}=g_{0}^{\ms{sk}^i}$ and $g_{1,i}=g_{1}^{\ms{sk}^i}$ for $i\in\Z$. Recall that since $\calA$ is an algebraic adversary, every group element it sends is accompanied by a representation in terms of the previous group elements it has seen. 
Just as we did in the proof of Theorem~\ref{thm:soundness}, instead of punching cards by raising $p_0^\ms{sk}$ and $p_1^\ms{sk}$, we examine the algebraic representation of $p_0,p_1$ submitted by the adversary and replace each instance of $g_{0,i}$ or $g_{1,i}$ with $g_{0,i+1}$ or $g_{1,i+1}$, respectively. Also, whenever the adversary $\calA$ makes a call to the oracles $H_j$ (for $j\in{0,1}$) on a previously unqueried point $u$, the challenger samples $r\getsr\Z_q$ and sets $H_j(u)\gets g_{j,1}^r$. 
\end{itemize}

From $H4$, we can prove that any algebraic adversary $\calA$ who wins the soundness game can be used by an algorithm $\calB$, described below, to break the $N$-discrete log assumption in either $G_0$ or $G_1$. Algorithm $\calB$ plays the role of the adversary in the asymmetric $N$-discrete log game while simultaneously playing the role of the challenger in $H4$. Algorithm $\calB$ simulates $H4$ exactly to $\calA$, except that it uses the asymmetric $N$-discrete log challenge messages $g_0^x, g_0^{x^2}, ..., g_0^{x^N}, g_1^x, g_1^{x^2}, ..., g_1^{x^N}$ as the values of $g_{0,i}$ and $g_{1,i}$. That is, $g_{0,i}=g_0^{x^i}$ and $g_{1,i}=g_1^{x^i}$. Moreover, it sets $\ms{pk}\gets (g_{0,1}, g_{1,1})$ in the setup phase. Observe that all $g_{0,i}$ and $g_{1,i}$ are distributed identically as in $H4$, so this is a perfect simulation of $H4$ with $x$ playing the role of $\ms{sk}$. The value of $N$ required in the assumption depends on the maximum number of sequential punches $\calA$ requests on the same group element.  

Now, if $\calA$ wins the soundness game, it means that $c_\ms{redeem}>c_\ms{punch}$. This, in turn, implies that there was some successful punch card redemption $\ms{ServerVerify}$ where the accepted value of $p$ had not been previously punched $n$ times between the two merged cards. Let $a$ and $b$ the number of times each of the two merged cards had been punched before redemption, so we have $a+b<n$. 

The successful verification requires that 
$$p = e(H_1(u)^{x^n}, H_2(u')) = e(g_{0,1}^{rx^n}, g_{1,1}^{r'}) = e(g_{0,1}^{rx^{a'}}, g_{1,1}^{r'x^{b'}})$$ 
for any $a',b'$ where $a'+b'=n$, and the algebraic adversary must give a representation of $p$ (it can include a pairing in this representation). It must be true that one of $a'$ or $b'$ is greater than $a$ or $b$, respectively, because $a+b<n$. Thus the representation of $p$ must not include either $g_{0,a'+1}^r$ or $g_{1,b'+1}^{r'}$ because one of those values will not have been given to $\calA$. This means we now have two different representations of one of these elements, which together yield a degree $a'+1$ or $b'+1$ equation in $x$. This equation can be solved for $x$ using standard techniques~\cite{shoupbook}, allowing $\calB$ to recover $x$ and win the asymmetric $N$-discrete log game. 
\end{proof}

\section{Implementation and Evaluation}\label{eval}

\begin{table*}\centering \normalsize
\begin{tabular}{lllllll}\toprule
&\ms{ServerSetup}&\ms{Issue}&\ms{ServerPunch}&\ms{ClientPunch}&\ms{ClientRedeem}&\ms{ServerVerify}\\\midrule
Computation Time (ms)& 0.019 & 0.304 & 0.134 & 4.314 & 0.890 & 0.064\\
Data Sent (Bytes)& 32 & 0 & 128 & 32 & 64 & 0\\\bottomrule
\end{tabular}
\vspace{.5em}
\caption{Computation and communication costs for our main punch card scheme. }
\label{fig:mainperf}
\end{table*}

\begin{table*}\centering \normalsize
\begin{tabular}{lllllll}\toprule
&\ms{ServerSetup}&\ms{Issue}&\ms{ServerPunch}&\ms{ClientPunch}&\ms{ClientMergeRedeem}&\ms{ServerVerify}\\\midrule
Computation Time (ms)& 1.09 & 34.97 & 4.33 & 137.79 & 36.43 & 4.00\\
Data Sent (Bytes)& 144 & 0 & 496 & 144 & 640 & 0\\\bottomrule
\end{tabular}
\vspace{.5em}
\caption{Computation and communication costs for our mergeable punch card scheme using pairings. }
\label{fig:mergeperf}
\end{table*}

We implemented our main punch card scheme from Section~\ref{punchcard} as well as the mergeable punch card scheme from Section~\ref{merge}. Our implementation is written in Rust with a Java wrapper to run the Rust code on Android devices. The implementation of the main punch card scheme relies on the \texttt{curve25519-dalek}~\cite{curve25519-dalek} crate which implements curve25519~\cite{bern06}, and the mergeable punch card scheme uses the \texttt{pairing-plus}~\cite{pairing-plus} crate, which provides an implementation of BLS12-381 curves~\cite{BLS12-381}. Our implementation and raw evaluation data are available at 
\url{https://github.com/SabaEskandarian/PunchCard}.

We carried out our evaluation with the client running on a Google Pixel (first generation) phone and the server running on a laptop with an Intel i5-8265U processor @ 1.60GHz. 
All data reported on our scheme comes from an average of at least 100 trials. \ms{ServerVerify} was run with $n=10$ punches on each redeemed punch card and an empty database \ms{DB} of used cards. We repeated the test of the main scheme with a database of 1,000,000 used cards and saw no significant difference between that and the test with an empty database, leading us to conclude that the hash table lookup does not dominate the cost of the \ms{ServerVerify} algorithm. 

Figure~\ref{fig:mainperf} shows the running time of each of the algorithms in our main punch card construction as well as the amount of in-protocol data sent by the party running the algorithm in a punch card system. The data sent in \ms{ServerSetup} refers to the size of the public key which must be communicated to clients. Observe that we do not require the client to communicate any data in order for a new punch card to be issued. The most costly operation, punching a card, requires less than 5ms between the client and server combined, and all other operations require less than 1ms. Communication is under 200 Bytes for all operations. 

Figure~\ref{fig:mergeperf} shows the same information for the mergeable punch card scheme. The mergeable scheme runs considerably more slowly than the main scheme. This is the result of 1) more work being required in the mergeable scheme, 2) group operations being more costly in pairing groups, and 3) the heavily optimized library used for curve25519 (\texttt{curve25519-dalek}) in the implementation of the main scheme. Group elements in pairing groups are also larger than in curve25519. In curve25519, the size of a group element $g\in G$ is 32 Bytes, but in the BLS12-381 curves, $g_0\in G_0$ requires 48 Bytes, $g_1\in G_1$ requires 96 Bytes, and $g_T\in G_T$ requires 576 Bytes.  

\paragraph{\newrev{Comparison to Blind Signatures}}\newrev{. Before presenting our empirical comparison to prior work, we will begin with a rough comparison of our approach to the blind signature-based scheme sketched in Section~\ref{intro}. Since exponentiations are the most expensive operation in both protocols, we will use them to roughly compare performance. Assuming the blind signature scheme is instantiated with blind Schnorr signatures~\cite{Schnorr89,Schnorr91}, it incurs only a few exponentiations per hole punch, comparable to our scheme. However, the client and server storage, as well as redemption time, will be proportional to the number of hole punches required to redeem a card. Whereas our scheme can be implemented with only one exponentiation required for a server to verify a punch card, the blind signature scheme would require two exponentiations per hole punch to verify each Schnorr signature. This means that our scheme outperforms the blind signature approach during redemption, even when only one hole punch is required to redeem a card. Our communication costs for redemption are also equivalent to the size of a single Schnorr signature, so communicating a single hole punch for redemption in the blind signature scheme, which involves sending a Schnorr signature and corresponding client-chosen secret, would already incur more communication than our scheme does regardless of the number of punches on a card. } 

\newrev{Moreover, during card redemption in the blind signature-based scheme, the server checks that all the secrets presented by the client have been signed and then adds each one to a database of previously redeemed secrets, checking for double-spending as it goes. As such, a card that requires 10 punches to be redeemed would introduce an order of magnitude difference between the two approaches, and this gap would increase linearly as the number of punches required grows. }

\newrev{The situation changes somewhat when comparing with our mergeable scheme, where the blind signature scheme does not require any modification because it is trivially mergeable. In this setting, the server-side database storage costs of our mergeable scheme are still smaller as long as a punch card requires more than two punches to be redeemed. On the other hand, hole punches will be noticeably slower because each punch operation in the mergeable scheme requires twice as many exponentiations as the main scheme. The fact that the mergeable scheme also requires using groups $G_0$ and $G_1$ that support a pairing means that each exponentiation will also be slower. In our experiments, $G_0$ exponentiations were about $5\times$ slower and $G_1$ exponentiations were about $16\times$ slower than in curve25519. }

\newrev{Redemption in the mergeable scheme also requires a pairing, which in our experiments were about $40\times$ slower than an exponentiation in curve25519. Despite exponentiations and pairings being significantly slower than those used in the blind signature scheme, we estimate, based on the relative performance of exponentiations and pairings in the different curves, that redemption in our mergeable scheme will still be faster on the server when a punch card requires $23$ or more punches to redeem, and will require less communication when a card requires $8$ or more punches to redeem.}

\begin{table*}\centering \normalsize
\begin{tabular}{llll}\toprule
&Issuing a Card&Punching a Card&Redeeming a Card\\\midrule
BBA+ scheme & 115.27 & 385.61 & 375.73\\
UACS scheme & 86 & 127 & 454\\
Bobolz~\etal scheme & 130 & 64 & 1254\\
BBW Scheme & 52 & 62 & 122 \\
Our main scheme & 0.304 (171.1$\times$ faster) & 4.448 (13.9$\times$ faster) & 0.954 (127.9$\times$ faster)\\
Our mergeable scheme& 34.97 (1.5$\times$ faster) & 142.12 & 40.43 (3.0$\times$ faster) \\\bottomrule
\end{tabular}
\vspace{.5em}
\caption{Computation time (in milliseconds) for our schemes and prior work. Speedups shown in parentheses refer to improvement over best prior work.}
\label{fig:computation}
\end{table*}

\begin{table*}\centering \normalsize
\begin{tabular}{llll}\toprule
&Issuing a Card&Punching a Card&Redeeming a Card\\\midrule
BBA+ scheme & 992 & 4048 & 3984\\
BBW scheme & 1005 & 1745 & 3502 \\
Our main scheme & 0 & 160 (10.9$\times$ reduction) & 64 (54.7$\times$ reduction) \\
Our mergeable scheme& 0 & 640 (2.7$\times$ reduction) & 640 (5.5$\times$ reduction)\\\bottomrule
\end{tabular}
\vspace{.5em}
\caption{Communication (in Bytes) in our schemes and prior works that report communication costs. Our schemes incur no communication to issue a new card and achieve order of magnitude improvements for other operations. The pairing-based scheme requires more communication because pairing group elements are larger and punching a card requires twice as many elements communicated. }
\label{fig:communication}
\end{table*}

\paragraph{Comparison to prior work}. We \newrev{now empirically} compare our punch card scheme to the \newrev{recent} loyalty systems BBA+~\cite{bbaplus}, UACS~\cite{uacs}, Bobolz~\etal~\cite{uacsplus}, and BBW~\cite{bbw}. We do not compare to the original BBA work~\cite{bba} because its performance is strictly worse than the works to which we do compare. 

We use the performance numbers reported by each prior work to which we compare. Performance numbers for UACS and Bobolz~\etal, were also recorded with a Google Pixel phone but used a computer with a stronger i7 processor. BBA+ only reports the client-side cost of each of its protocols and uses a OnePlus 3 phone. BBW uses an unspecified smartphone with a Snapdragon 845 processor (used, e.g., in the Pixel 3). In order to better capture the total cost of using each approach, we combine client and server costs to give the overall computation cost of each scheme. However, the distribution of cost between the client and server is similar for all works, with the mobile device incurring most of the computation cost. Since BBW measured its server-side costs on the same smartphone instead of an external server, we only count their client-side costs in our evaluation to ensure fairness. Moreover, BBW reports different configurations to achieve the best computation and communication costs, so we use the best reported number for each comparison. 

Figure~\ref{fig:computation} compares the performance numbers of our schemes against those of prior work. Our main scheme issues a card $171.1\times$ faster, punches a card $13.9\times$ faster, and redeems a card $127.9\times$ faster than the best prior work, BBW. The performance improvement comes from removing the reliance on pairings (in all works except BBW) and significantly reducing the number and complexity of zero knowledge proofs required in each operation compared to all prior works. \rev{We are able to do this by tailoring our solution more narrowly to supporting punch cards. Prior works require heavier zero-knowledge proofs to handle broader use cases, e.g, adding negative points, so focusing on punch cards allows us to do away with much of the complexity of prior solutions.}

Our mergeable punch card scheme outperforms prior work in almost every category despite using pairings. The BBW scheme punches cards $2.2\times$ faster than our scheme, but in return our scheme issues cards $1.5\times$ faster and redeems cards $3\times$ faster. An important difference to point out between our implementation and most prior work is that while our implementation is done with BLS12-381 curves, which provide 128 bits of security, all prior works except BBW use BN curves~\cite{BN05,bncurves} which only provide 100 bits of security~\cite{TB16,pairingcurves}.

Figure~\ref{fig:communication} compares the communication costs of our schemes with BBA+ and BBW, the only prior works to report the communication costs incurred by their implementation. Unlike all prior work, our scheme requires \emph{no communication} to issue a new card, and card punching and redemption require $10.9\times$ and $54.7\times$ less communication, respectively, than the best prior work. For the mergeable scheme, the improvements are reduced to $2.7\times$ and $5.5\times$, but even this scheme requires significantly less communication.

\section{Extensions} \label{extensions}

We now briefly discuss extensions to our main punch card scheme that can allow it to be used in a wider variety of applications. 

\paragraph{Multi-punches}. 
Some loyalty programs sometimes offer extra punches on their punch cards as a special promotion. Others don't use a punch card at all, opting instead for a system where different transactions earn varying numbers of points. Our punch card scheme can easily be extended to handle these situations by having the server raise $p$ to $\ms{sk}^t$, where $t$ is the number of points being awarded for a given transaction. 

Unfortunately, this kind of multi-punch raises a new security question. Most punch card schemes offer a fixed value $n$ at which point a card can be redeemed for some benefit, or perhaps a few values at which different kinds of rewards are unlocked. But the possibility of gaining more than one punch with a given transaction introduces the potential for a client to ``overshoot'' the required number of points. This does not pose an issue for functionality, because the client can just redeem a card with $n'>n$ punches and perhaps even get a new card with the remaining balance. However, this might introduce a privacy issue because the redemption reveals the total number of punches on a card, which is no longer always the exact same value for all clients. One way to eliminate this problem is to have the server send all possible values $p^\ms{sk}, p^{\ms{sk}^2},...,p^{\ms{sk}^t}$ when punching a card. This works well for settings where $t$ is small, e.g., a double-punch promotion. We leave the problem of an efficient solution for large $t$ for future work. 

\paragraph{Managing used card database size}. 
Our punch card scheme requires keeping a database \ms{DB} of used punch card secrets $u$, stored in a hash table in our implementation. While this does not pose a performance problem because of the amortized constant time lookup in the hash table, the storage cost increases over time. Although at 128 bits per secret, it would take a long time for storage costs to become prohibitive, a high-volume punch card program may wish for a plan to eventually remove old punch cards from the database without allowing double spending. 

One way to help reduce the long-term storage requirement is by adding extra information into the secret $u$. Since $u$ is ultimately passed through a hash function modeled as a random oracle, adding structured information before the random bits makes no difference in the security of the scheme (unless the structured information itself leaks something). Clients can be required to add an expiration date to the beginning of $u$. Then the card redemption would check whether the card being used is expired or not. To encourage clients to pick reasonable expiration dates, cards with expiration dates too far in the future could be rejected as well. Expiration dates used in $u$ could be standardized, e.g., to the first day of a given year, to prevent the date itself from leaking too much information about an individual customer's shopping habits. 

\paragraph{\rev{Proof of Redemption}}. \rev{Our scheme can easily be modified to allow servers to prove that a punch card has previously been redeemed, e.g., when declining to accept a customer's otherwise valid card. This can be achieved by having the card secret $u$ itself be derived as the hash of another \emph{redemption secret} \ms{rs} generated by the client during issuance and sent to the server at the end of redemption. To prove that a card has previously been redeemed, the server sends back \ms{rs} to the client before this last step. Including this functionality adds an additional hash to card issuance and redemption as well as an additional communication round during redemption, but it may be desirable in situations where clients have reason to believe that a server may try to dishonestly reject valid punch cards.}

\paragraph{Private ticketing}. 
Our punch card scheme can also be viewed as a scheme for private ticketing, or, more generally, as a one-time use keyed-verification anonymous credential~\cite{CMZ14}. To issue a ticket, the client generates a new punch card, and the server punches it. A ticket can reflect additional information (e.g., if a train ticket is first class or coach, which transit zones a ticket is valid for, etc.) by the number of punches added to the ticket. To record multiple pieces of information on the same ticket, the random oracle $H$ can be used to generate multiple group elements from $u$, each of which can hold a different number of punches. Since the punches cannot be linked to their redemption, a client can later present the ticket without linking it to the issuance process.

\section{Related Work}\label{related}

In principle, the problem of privacy-preserving punch cards can be approached with a number of different techniques. One approach to this problem is via standard anonymous credential techniques~\cite{Chaum85,CL01,CL04}. Ecash systems~\cite{CHL05,BCKL09} or even the uCentive system~\cite{ucentive}, which is specifically designed for loyalty programs, can be used to give a customer an unlinkable token for each purchase. However, storage and computation costs to hold and redeem a token in these systems must be linear in the number of ``hole punches'' a customer acquires. Cryptographic techniques that could be used to accumulate constant-sized tokens do exist (e.g.,~\cite{GCZY13, RBHP15}), but they rely on pairings and incur costs higher than the protocols presented in this work. 

The line of work which this paper follows begins with the Black Box Accumulation of Jager and Rupp~\cite{bba}. This work introduced the notion of a constant-sized wallet that could accumulate points without being linkable, but used heavier tools than more recent work and offered weaker security. In particular, although accumulation of individual points was unlinkable in this scheme, the processes of issuing and removing points were linkable. 

This limitation was removed in subsequent works BBA+~\cite{bbaplus} and UACS~\cite{bbaplus}, which both strengthened security definitions and improved performance compared to BBA. These systems also added new features, such as partial spending of points, which enables new applications not covered by the original BBA model. Bobolz~\etal~\cite{uacsplus} further refine UACS, dramatically reducing the time to add points by pushing most costly operations into the redemption phase. Moreover, their scheme can recover from failures that occur mid-protocol. Our scheme trivially satisfies this requirement because the protocol only requires one message from each party and improves the performance of issuing, punching, and redeeming compared to prior work. 

BBW~\cite{bbw}, much like this work, improves on the performance of BBA+ by removing the reliance on pairings and using zero-knowledge range proofs, instantiated with Bulletproofs~\cite{bulletproofs}, instead. Our work focuses more narrowly on the punch card scenario, without supporting addition of negative points, but it shows how we can achieve order of magnitude improvements even over BBW and dispense with all but the simplest of proofs in this setting.

\section{Conclusion}\label{concl}

We have presented a new scheme for punch card loyalty programs that significantly outperforms all prior work both quantitatively and qualitatively. Our scheme does not require any server interaction for a client to receive a punch card, does not require pairings, and outperforms prior work in card issuance, punching, and redemption by $171\times$, $14\times$, and $128\times$ respectively, strictly dominating the performance of all prior solutions to this problem. We have also shown several extensions to our main scheme, including a modified protocol that allows merging punch cards (using pairings) that still outperforms prior work. Our implementation is open source and available at 
\url{https://github.com/SabaEskandarian/PunchCard}.

\section*{Acknowledgments}

I would like to thank Dan Boneh for several helpful conversations, as well as my shepherd Dan Roche and the anonymous reviewers for helpful suggestions that improved the final draft of this paper. This work was funded by NSF, DARPA, a grant from ONR, and the Simons Foundation. Opinions, findings and conclusions or recommendations expressed in this material are those of the authors and do not necessarily reflect the views of DARPA.


\bibliographystyle{plain}
\bibliography{paper}


\end{document}